%% file: paper.tex
\documentclass[runningheads]{llncs}
\usepackage[utf8]{inputenc}
\usepackage{graphicx}
\usepackage{amsmath,amssymb}
\usepackage{multirow}
\usepackage{array}
\usepackage{float}
\usepackage{wrapfig}
\usepackage{algorithm,algpseudocode}
\usepackage{enumitem}
\newcommand{\sol}{\textsc{sol}}

\begin{document}

\title{Concave connection cost Facility Location and the Star Inventory Routing problem. \thanks{Authors were supported by the NCN grant number 2015/18/E/ST6/00456}}

% \authorrunning{J. Byrka et al.}
\titlerunning{Concave connection cost Facility Location and \textsc{SIRPFL}}

\author{Jarosław Byrka\inst{1}\orcidID{0000-0002-3387-0913} \and \\
Mateusz Lewandowski\inst{1}\orcidID{0000-0003-2912-099X}}

\authorrunning{J. Byrka and M. Lewandowski}
% \authorrunning{XXX and YYY}

\institute{Institute of Computer Science, University of Wrocław, Poland}

\maketitle

\begin{abstract}
We study a variant of the \emph{uncapacitated facility location} (\textsc{ufl}) problem, where connection costs of clients are defined by (client specific) concave nondecreasing functions of the connection distance in the underlying metric. A special case capturing the complexity of this variant is the setting called \emph{facility location with penalties} where clients may either connect to a facility or pay a (client specific) penalty. 

We show that the best known approximation algorithms for \textsc{ufl} may be adapted to the concave connection cost setting. The key technical contribution is an argument that the JMS algorithm for \textsc{ufl} may be adapted to provide the same approximation guarantee for the more general concave connection cost variant.

We also study the \emph{star inventory routing with facility location} (\textsc{sirpfl}) problem that was recently introduced by Jiao and Ravi,
which asks to jointly optimize the task of clustering of demand points with the later serving of requests within created clusters. We show that the problem may be reduced to the concave connection cost facility location and substantially improve the approximation ratio for all three variants of \textsc{sirpfl}.

\keywords{Facility location \and Inventory Routing \and Approximation}
	\end{abstract}

\section{Introduction}
    \input{introduction}

\section{JMS with penalties}
    \input{jms}

\section{Combining algorithms for FLP}
    \input{combining}
\section{Solving NCC-FL with algorithms for FLP}
    \input{reduction}
\section{Approximation algorithms for SIRPFL}
    \input{IRP}    
%\section*{Concluding remarks }
%We believe that our results could be useful in other contexts as the abstraction provided by our \textsc{ncc-fl} problem is quite general. In particular an algorithm for \textsc{ncc-fl} may potentially be used as a subroutine when ...

\bibliographystyle{splncs04}
\bibliography{references}

\end{document}

%% file: introduction.tex
The uncapacitated facility location (\textsc{ufl}) problem has been recognized by both theorists and practitioners as one of the most fundamental problems in combinatorial optimization. In this classical NP-hard problem, we are given a set of facilities $F$ and a set of clients $C$. We aim to open a subset of facilities and connect each client to the closest opened facility. The cost of opening a facility $i$ is $f_i$ and the cost of connecting the client $j$ to facility $i$ is the distance $d_{j,i}$. The distances $d$ are assumed to define a symmetric metric. We want to minimize the total opening costs and connection costs.

The natural generalization is a variant with penalties. For each client $j$, we are given its penalty $p_j$. Now, we are allowed to reject some clients, i.e., leave them unconnected and pay some fixed positive penalty instead. The objective is to minimize the sum of opening costs, connection costs and penalties. We call this problem facility location with penalties and denote as \textsc{flp}.

We also study inventory routing problems that roughly speaking deal with scheduling the delivery of requested inventory to minimize the joint cost of transportation and storage subject to the constraint that goods are delivered in time. The approximability of such problems have been studied, see e.g.,~\cite{nagarajan2016approximation}.

Recently, Jiao and Ravi~\cite{IRP:WADS} proposed to study a combination of inventory routing and facility location. The obtained general problem appears to be very difficult, therefore they focused on a special case, where the delivery routes are stars. They called the resulting problem the Star Inventory Routing Problem with Facility Location (\textsc{sirpfl}). Formally, the problem can be described as follows. We are given a set of clients $D$ and facility locations $F$ with opening costs $f_i$ and metric distances $d$ as in the \textsc{ufl} problem. Moreover we are given a time horizon $1, \dots, T$ and a set of demand points $(j,t)$ with $u^j_t$ units of demand for client $j \in D$ due by day $t$. Furthermore, we are given holding costs $h^j_{s,t}$ per unit of demand delivered on day $s$ serving $(j,t)$. The goal is to open a set of facilities, assign demand points to facilities, and plan the deliveries to each demand point from its assigned facility. For a single delivery on day $t$ from facility $i$ to client $j$ we pay the distance $d_{j,i}$. The cost of the solution that we want to minimize is the total opening cost of facilities, delivery costs and the holding costs for early deliveries.

The above \textsc{sirpfl} problem has three natural variants. In the \emph{uncapacitated} version a single delivery can contain unlimited number of goods as opposed to \emph{capacitated} version, where a single order can contain at most $U$ units of demand. Furthermore, the capacitated variant can be \emph{splittable}, where the daily demand can be delivered across multiple visits and the \emph{unsplittable}, where all the demand $(j,t)$ must arrive in a single delivery (for feasibility, the assumption is made that a single demand does not exceed the capacity $U$).

\subsection{Previous work}
	The metric \textsc{ufl} problem has a long history of results~\cite{GUHA1999228,ChudakS03,JainV01,MahdianYZ02,Sviridenko02}. The current best approximation factor for \textsc{ufl} is $1.488$ due to Li~\cite{li20111}. This is done by combining the bifactor\footnote{Intuitively a bifactor $(\lambda_f, \lambda_c)$ means that the algorithm pays at most $\lambda_f$ times more for opening costs and $\lambda_c$ times more for connection cost than the optimum solution.} $(1.11,1.78)$-aproximation algorithm by Jain et al.~\cite{Jain:2003:GFL:950620.950621} (JMS algorithm) with the LP-rounding algorithm by Byrka and Aardal~\cite{byrka2010optimal}. The analysis of the JMS algorithm crucially utilizes a \emph{factor revealing LP} by which the upper bound on the approximation ratio of the algorithm is expressed as a linear program. For the lower bounds, Sviridenko~\cite{sviridenkohardness} showed that there is no better than $1.463$ approximation for metric $\textsc{ufl}$ unless $\text{P} \neq \text{NP}$.

	The above hardness result transfers to the penalty variant as \textsc{flp} is a generalization of \textsc{ufl}. For approximation, the long line of research~\cite{CharikarKMN01,XuX05,XuX09,GeunesLRS11,LiDXX15} stopped with the current best approximation ratio of $1.5148$ for \textsc{flp}. It remains open\footnote{Qiu and Kern~\cite{qiu2016factor} claimed to close this problem, however they withdrawn their work from arxiv due to a crucial error.}, whether there is an algorithm for \textsc{flp} matching the factor for classical \textsc{ufl} without penalties.

	For the \textsc{sirpfl} problem, Jiao and Ravi~\cite{IRP:WADS} gave the $12$, $24$ and $48$ approximation algorithms for uncapacitated, capacitated splittable and capacitated unsplittable variants respectively using LP-rounding technique.

\subsection{Nondecreasing concave connection costs}
	We propose to study a natural generalization of the \textsc{flp} problem called per-client nondecreasing concave connection costs facility location~(\textsc{ncc-fl}). The set up is identical as for the standard metric \textsc{ufl} problem, except that the connection cost is now defined as a function of distances. More precisely, for each client $j$, we have a nondecreasing concave function $g_j$ which maps distances to connection costs. %The goal is to open a subset of facilities and connect each client to minimize total opening costs and connection costs.{}
	We note the importance of concavity assumption of function $g_j$. Dropping this assumption would allow to encode the set cover problem similarly to the non-metric facility location, rendering the problem hard to approximate better than within $\log n$.

	As we will show, the \textsc{ncc-fl} is tightly related to \textsc{flp}. We will also argue that an algorithm for \textsc{ncc-fl} can be used as a subroutine when solving \textsc{sirpfl} problems. Therefore it serves us a handy abstraction that allows to reduce the \textsc{sirpfl} to the \textsc{flp}. 
	%More broadly, \textsc{ncc-fl} can be also used to solve other facility location problems in which every client has to additionally solve some subproblem that is parametrized by the distance to the closest opened facility.
\begin{table}[h]
		\centering
		\begin{tabular}{r|c|c}
		                        				 & \textbf{previous work} & \textbf{our results} \\ \hline
		\textsc{flp}                      		 & 1.5148~\cite{LiDXX15}  & 1.488       \\ \hline
		\textsc{ncc-fl}             		     & -                      & 1.488       \\ \hline
		uncapacitated \textsc{sirpfl}            & 12~\cite{IRP:WADS}     & 1.488       \\ \hline
		capacitated splittable \textsc{sirpfl}   & 24~\cite{IRP:WADS}     & 3.236       \\ \hline
		capacitated unsplittable \textsc{sirpfl} & 48~\cite{IRP:WADS}     & 6.029       \\
		\end{tabular}
		\vspace{10pt}
		\caption{Summary of improved approximation ratios}
		\label{table:results}
	\end{table}	
\subsection{Our results}
	We give improved approximation algorithms for \textsc{flp} and all three variants of \textsc{sirpfl} (see Table~\ref{table:results}). Our work closes the current gap between classical facility location and \textsc{flp}. %The improvement for \textsc{sirpfl} problems is also substantial. 
	More precisely, our contributions are as follows:
	\begin{enumerate}
		\item We adapt the JMS algorithm to work for the penalized variant of facility location. The technical argument relies on picking a careful order of the clients in the factor revealing program and an adequate reduction to the factor revealing program without penalties.
		\item Then, we combine the adapted JMS algorithm with LP rounding to give the $1.488$-approximation algorithm for \textsc{flp}. Therefore we match the best known approximation algorithm for \textsc{ufl}.
		\item We show a reduction from the \textsc{ncc-fl} to \textsc{flp} which results in a $1.488$-approximation algorithm for \textsc{ncc-fl}.
		\item We cast the \textsc{sirpfl} as the \textsc{ncc-fl} problem, therefore improving approximation factor from $12$ to $1.488$.
		\item For the capacitated versions of \textsc{sirpfl} we are also able to reduce the approximation factors from $24$ (for splittable variant) and $48$ (for unsplittable variant) down to $3.236$ and $6.029$ respectively.
	\end{enumerate}
The results from points 2 and 3 are more technical and follow from already known techniques, we therefore only sketch their proofs in Sections 3 and 4, respectively. The other arguments are discussed in detail.

%% file: jms.tex
Consider Algorithm~\ref{alg:modified-JMS}, a natural analog of the JMS algorithm for penalized version. The only difference to the original JMS algorithm is that we simply freeze the budget $\alpha_j$ of client $j$ whenever it reaches $p_j$. 
%We provide the detailed description in the Algorihm~\ref{alg:modified-JMS}. 
For brevity, we use notation $[x]^+ = \max\{x, 0\}$.

    \begin{algorithm}[h!]
        \caption{Penalized analog of JMS}
        \label{alg:modified-JMS}
        \begin{algorithmic}[1]
            \item Set budget $\alpha_j := 0$ for each client $j$. Declare all clients active and all facilities unopened. At every moment each client $j$ offers some part of its budget to facility~$i$. The amount offered is computed as follows:
            \begin{enumerate}[label=(\roman*)]
                \item if client $j$ is not connected: $[\alpha_j - d_{i,j}]^+$
                \item if client $j$ is connected to some other facility $i'$: $[d_{i',j} - d_{i,j}]^+$
            \end{enumerate}
            \item While there is any active client:
                \begin{itemize}
                    \item simultaneously and uniformly increase the time $t$ and budgets $\alpha_j$ for all active clients until one of the three following events happen:
                    \begin{enumerate}[label=(\roman*)]
                        \item \emph{facility opens}: for some unopened facility $i$, the total amount of offers from all the clients (active and inactive) is equal to the cost of this facility.
                        In this case open facility $i$, (re-)connect to it all the clients with positive offer towards $i$ and declare them inactive.
                        \item \emph{client connects}: for some active client $j$ and opened facility $i$, the budget $\alpha_j = d_{i,j}$. In this case, connect a client $j$ to facility $i$ and deactivate client $j$.
                        \item \emph{potential runs out}: for some active client $j$, its budget $\alpha_j = p_j$. In this case declare $j$ inactive.
                    \end{enumerate}
                \end{itemize}
            \item Return the set of opened facilities. We pay penalties for clients that did not get connected.
        \end{algorithmic}
    \end{algorithm}
    Observe that in the produced solution, clients are connected to the closest opened facility and we pay the penalty if the closest facility is more distant then the penalty.

    Note that Algorithm~\ref{alg:modified-JMS} is exactly the same as the one proposed by Qiu and Kern~\cite{qiu2016factor}. In the next section we give the correct analysis of this algorithm.
\subsection{Analysis: factor-revealing program}
    We begin by introducing additional variables $t_j$ to Algorithm~\ref{alg:modified-JMS}, which does not influence the run of the algorithm, but their values will be crucial for our analysis. Initially set all variables $t_j := 0$. As Algorithm~\ref{alg:modified-JMS} progresses, increase variables $t_j$ simultaneously and uniformly with global time $t$ in the same way as budgets $\alpha_j$. However, whenever \emph{potential runs out} for client $j$, we do not freeze variable $t_j$ (as opposed to $\alpha_j$), but keep increasing it as time $t$ moves on. For such an inactive client $j$, we will freeze $t_j$ at the earliest time $t$ for which there is an opened facility at distance at most $t$ from $j$. For other active clients (i.e. the clients that did not run out of potential) we have $t_j = \alpha_j$.

    Observe now, that the final budget of a client at the end of the algorithm is equal to $\alpha_j = \min\{t_j, p_j\}$. We will now derive a factor revealing program and show that it upper-bounds the approximation factor.
    \begin{theorem}
        \label{thm:program-with-penalites}
        Let $\lambda_f \geq 1$. Let also $\lambda_c = \sup_k{o_k}$, where $o_k$ is the value of the following optimization program $P(k)$.
        \begin{align}
            \max \hspace{5pt} & \frac{\sum_i{\min\{t_i,p_i\}} - \lambda_f f}{\sum_i{d_i}} & (P(k)) \nonumber\\
            \text{s.t.} &                                    \nonumber \\
            & \sum_{i=1}^{l-1} \left[\min\{r_{i,l}, p_i\} - d_i\right]^+   +   \sum_{i = l}^{k} \left[\min\{t_l, p_i\} - d_i\right]^+   \leq f & l \in [k] \label{constraint:opening-start}\\
            & t_i \leq t_{i+1}                      & i \in [k-1] \label{constraint:monotone}\\
            & r_{j,i} \geq r_{j, i+1}               & 1 \leq j < i < k \label{constraint:r-monotone}\\
            & t_i \leq r_{j,i} + d_i + d_j          & 1 \leq j < i \leq k \label{constraint:metric}\\
            & r_{i,l} \leq t_i                        & 1 \leq i < l \leq k \label{constraint:r-vs-t}\\
            & p_i \geq d_i                        & 1 \leq i \leq k \label{constraint:p-vs-d} \\
            & t_i \geq 0, d_i \geq 0, r_{j,i} \geq 0, f \geq 0   & 1 \leq j \leq i \leq k \label{constraint:nonneg1}\\
            & p_i \geq 0                            \label{constraint:nonneg2}
        \end{align}
        Then, for any solution $S$ with facility cost $F_S$, connection cost $D_S$ and penalty cost $P_S$ Algorithm~\ref{alg:modified-JMS} returns the solution of cost at most $\lambda_f F_S + \lambda_c (D_S + P_S)$.
    \end{theorem}
    \begin{proof}
        It is easy to see that Algorithm~\ref{alg:modified-JMS} returns a solution of cost equal to the total budget, i.e., $\sum_i{\alpha_i} = \sum_i{\min\{t_i,p_i\}}$.
        To show a bifactor $(\lambda_f, \lambda_c)$ we fix $\lambda_f$ and consider the $\lambda_f \cdot F_S$ of the budget as being spent on opening facilities and ask how large the resulting $\lambda_c$ can become. Therefore we want to bound $\frac{\sum_i{\alpha_i} - \lambda_f F_S}{D_S + P_S}$.

        Observe that solution $S$ can be decomposed into set $X$ of clients that are left unconnected and a collection $\mathcal{C}$ of stars. Each star $C \in \mathcal{C}$ consist of a single facility and clients connected to this facility (clients for which this facility was closest among opened facilities). Let $d(C)$ and $f(C)$ denote the connection and facility opening cost of the star $C$ respectively. We have to bound the following:
        \begin{align*}
        \frac{\sum_i{\alpha_i} - \lambda_f F_S}{D_S + P_S} &=
        \frac{\sum_{C \in \mathcal{C}}\left(\sum_{j\in C}{\alpha_j} - \lambda_f f(C)\right) + \sum_{j\in X} \alpha_j}{\sum_{C\in \mathcal{C}} d(C) + P_S} \\
        &\leq \frac{\sum_{C \in \mathcal{C}}\left(\sum_{j\in C}{\alpha_j} - \lambda_f f(C)\right) + P_S}{\sum_{C\in \mathcal{C}} d(C) + P_S} \\
        &\leq \max_{C\in \mathcal{C}}\frac{\sum_{j\in C}{\alpha_j} - \lambda_f f(C)}{d(C)}
        \end{align*}
        where the first inequality comes from the fact that $\alpha_j \leq p_j$ for any $j$. The last inequality follows because we can forget about $P_S$ as the numerator is larger than the denominator.

        Therefore we can focus on a single star $C$ of the solution. Let $\mathcal{F}$ be the unique facility of this star and let $f$ be the opening cost of this facility. Let also $1, 2, \dots, k$ be the clients connected in this star to $\mathcal{F}$ and let $d_i$ be the distance between client $i$ and facility $\mathcal{F}$. We also assume that these clients are arranged in the nondecreasing order with respect to $t_i$. This is the crucial difference with the invalid analysis in~\cite{qiu2016factor}.

        For each $j<i$ we define $r_{j, i}$ as the distance of client $j$ to the closest opened facility in a moment just before $t_i$. The constraint (\ref{constraint:r-monotone}) is valid, because when time increases we may only open new facilities.

        To understand constraint (\ref{constraint:metric}), i.e., $t_i \leq r_{j,i} + d_i + d_j$ for $1\leq j < i \leq k$, consider the moment $t = t_i$. Let $\mathcal{F}'$ be the opened facility at distance $r_{j,i}$ from $j$. By the triangle inequality, the distance from $i$ to $\mathcal{F}'$ is at most $r_{j,i} + d_i + d_j$. The inequality follows from the way we defined $t_i$ (see Figure~\ref{fig:order-t}).

        \begin{figure}[t]
            \centering
            \includegraphics[width=.5\linewidth]{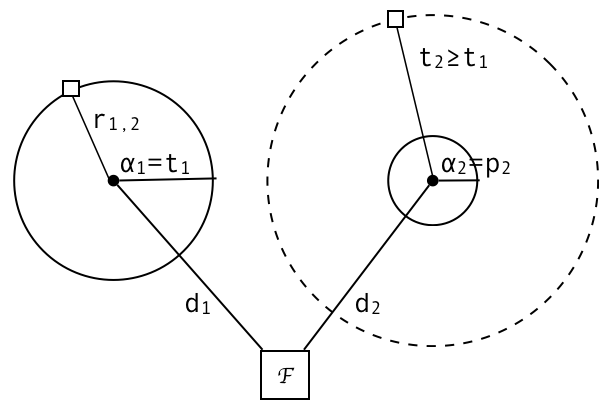}
            \caption{The ordering of the clients with respect to $t_i$ instead of $\alpha_i$.}
            \label{fig:order-t}
        \end{figure}

        Constraint (\ref{constraint:r-vs-t}), i.e. $r_{i,l} \leq t_i$ for $i < l$ follows, as client $i$ cannot be connected to a facility of distance larger than $t_i$. Moreover constraint (\ref{constraint:p-vs-d}) is valid, as otherwise solution $S'$ which does not connect client $i$ to facility $\mathcal{F}$ would be cheaper.
        
        Now we are left with justifying the opening cost constraints (\ref{constraint:opening-start}). Fix $l \in [k]$ and consider the moment just before $t_l$. We will count the contribution of each client towards opening $\mathcal{F}$. From Algorithm~\ref{alg:modified-JMS} the total contribution cannot exceed $f$. First, consider $i < l$. It is easy to see that if $i$ was already connected to some facility, then its offer is equal to $[r_{i,l} - d_i]^+$. Otherwise, as $t_i\leq t_l$, we know that $i$ already exhausted its potential, hence its offer is equal to $[p_i - d_i]^+$. Consider now $i \geq l$. From the description of Algorithm~\ref{alg:modified-JMS} and definition of $t_i$ its budget at this point is equal to $\min\{t_i, p_i\} \geq \min\{t_l, p_i\}$.

%        We justified all the constraints of the program $P(k)$. This finishes the proof.
    \qed\end{proof}
    Observe that $P(k)$ resembles the factor revealing program used in the analysis of the JMS algorithm for version without penalties~\cite{Jain:2003:GFL:950620.950621}. However $P(k)$ has additional variables $p$ and minimas.
    Consider now the following program $\hat{P}(k, m)$
    \begin{align}
        \max \hspace{5pt} & \frac{\sum_i{m_i t_i} - \lambda_f f}{\sum_i{m_i d_i}} & (\hat{P}(k, m)) \nonumber\\
        \text{s.t.} &                                    \nonumber \\
        & \sum_{i=1}^{l-1} m_i \left[r_{i,l} - d_i\right]^+   +   \sum_{i = l}^{k} m_i \left[t_l - d_i\right]^+   \leq f & l \in [k] \label{constraint:opening-final}\\
        & (\ref{constraint:monotone}), (\ref{constraint:r-monotone}), (\ref{constraint:metric}), (\ref{constraint:nonneg1}) \nonumber
    \end{align} where $m$ is the vector of $m_i$'s.
    We claim that by losing arbitrarily small $\epsilon$, we can bound the value of $P(k)$ by the $\hat{P}(k,m)$ for some vector $m$ of natural numbers. This is captured by the following theorem, where $v(\sol, P)$ denotes the value of the solution $\sol$ to program $P$.
    \begin{theorem}
        \label{thm:reduction-to-jms}
        For any $\epsilon > 0$ and any feasible solution $\sol$ of value $v(\sol, P(k))$ to the program $P(k)$, there exists a vector $m$ of natural numbers and a feasible solution $\overline{\sol}_{k,m}$ to program $\hat{P}(k,m)$ such that $v(\sol, P(k)) \leq v(\overline{\sol}_{k,m}, \hat{P}(k,m)) + \epsilon$.
    \end{theorem}
%    \begin{proof}
  %      See Section~\ref{sec:proof-reduction-to-jms}.
   % \qed\end{proof}
Before we prove the above theorem, we show that it implies the desired bound on the approximation ratio. Note that the program $\hat{P}(k,m)$ is similar to the factor revealing program in the statement of Theorem 6.1 in~\cite{Jain:2003:GFL:950620.950621} where their $k=\sum_{i=1}^k m_i$. The only difference is that it has additional constraints imposing that some clients are the same (we have $m_i$ copies of each client). However, this cannot increase the value of the program. Therefore, we obtain the same bi-factor approximation as the JMS algorithm~\cite{Jain:2003:GFL:950620.950621}.
    \begin{corollary}
    \label{corollary:bifactor}
        Algorithm~\ref{alg:modified-JMS} is a $(1.11,1.78)$-approximation algorithm\footnote{see e.g., Lemma~2 in~\cite{MahdianYZ02} for a proof of these concrete values of bi-factor approximation.} for the \textsc{ufl} problem with penalties, i.e., it produces solutions whose cost can be bounded by $1.11$ times the optimal facility opening cost plus $1.78$ times the sum of the optimal connection cost and penalties. 
    \end{corollary}
We are left with the proof of Theorem~\ref{thm:reduction-to-jms} which we give in the subsection below.
\subsection{Reducing the factor revealing programs}
\label{sec:proof-reduction-to-jms}
    \begin{proof}[Proof of Theorem~\ref{thm:reduction-to-jms}]
        First, we give the overview of the proof and the intuition behind it. We have two main steps:
        \begin{itemize}
            \item \textbf{Step 1 --- Getting rid of $p$ and minimas} \\
                We would like to get a rid of the variables $p_i$ and the minimas. To achieve this, we replace them with appropriate ratios $z_i$.
            \item \textbf{Step 2 --- Discretization} \\
                We then make multiple copies of each client. For some of them, we assign penalty equal to its $t_i$ and for others --- $d_i$. The portion of copies with positive penalty is equal to $z_i$.
        \end{itemize}
        
        For each step, we construct optimization programs and corresponding feasible solutions. The goal is to show, that in the resulting chain of feasible solutions and programs, the value of each solution can be upper-bounded by the value of the next solution.
        Formally, take any feasible solution $\sol = (t^*,d^*,r^*,p^*,f^*)$ to the program $P(k)$. 
        In the following, we will construct solutions $\sol_1, \sol_2$ and programs $P_1$ and $P_2$ such that
        $v(\sol, P(k)) \leq v(\sol_1, P_1) \leq v(\sol_2, P_2) + \epsilon \leq v(\overline\sol_{k,m}, \hat{P}(k,m)) + \epsilon$

        \textbf{Step 1 --- getting rid of $p$ and minimas}

        Define $z^*_i = \frac{\min\{t^*_i, p^*_i\} - d^*_i}{t^*_i - d^*_i}$ if $t^*_i - d^*_i >0$, and $z^*_i = 0$ otherwise. Observe that $z^*_i \in [0, 1]$ as $p^*_i \geq d^*_i$ by constraint~(\ref{constraint:p-vs-d}).
        We claim that $\sol_1 = (t^*,d^*,r^*,z^*,f^*)$ is a feasible solution to the following program $P_1(k)$
        \begin{align}
            \max \hspace{5pt} & \frac{\sum_i{d_i + z_i (t_i - d_i)} - \lambda_f f}{\sum_i{d_i}} & (P_1(k)) \nonumber\\
            \text{s.t.} &                                    \nonumber \\
            & \sum_{i=1}^{l-1} z_i \left[r_{i,l} - d_i\right]^+   +   \sum_{i=l}^{k} z_i \left[t_l - d_i\right]^+   \leq f & l \in [k] \label{constraint:opening2}\\
            & (\ref{constraint:monotone}), (\ref{constraint:r-monotone}), (\ref{constraint:metric}), (\ref{constraint:r-vs-t}), (\ref{constraint:nonneg1}) \nonumber \\
            & 0 \leq z_i \leq 1                            \label{constraint:nonneg-z}
        \end{align}
        and that the value of $\sol$ in $P(k)$ is the same as the value of $\sol_1$ in $P_1(k)$.
        The latter property follows trivially as $\sum_i{d^*_i + z^*_i (t^*_i - d^*_i)} = \sum_i{\min\{t^*_i,p^*_i\}}$.
        To show feasibility, we have to argue that (\ref{constraint:opening2}) is a valid constraint for $\sol_1$. To this end we will use the following claim.
        \begin{claim}
        \label{claim:minimum-and-z}
            For any $x\leq t^*_i$, we have that
            $$\left[\min\{x, p^*_i\} - d^*_i\right]^+ \geq z^*_i \left[x - d^*_i\right]^+ $$
        \end{claim}
        \begin{proof}
            Consider two cases:
            \begin{enumerate}
                \item $p^*_i \leq x$. In this case the left hand side is equal to $\left[p^*_i - d^*_i\right]^+$, while the right hand side is equal to $\frac{p^*_i-d^*_i}{t_i-d_i} \cdot \left[x - d^*_i\right]^+$. As $x\leq t^*_i$, the claim follows.
                \item $p^*_i > x$. In this case the left hand side is equal to $\left[x - d^*_i\right]^+$, while the right hand side is equal to $z^*_i \cdot \left[x - d^*_i\right]^+$. As $z^*_i \leq 1$, the claim follows.
            \end{enumerate}
        \qed\end{proof}
        Claim~\ref{claim:minimum-and-z} together with the fact that $r^*_{i,l} \leq t^*_i$ for $i<l$ (constraint (\ref{constraint:r-vs-t})) and $t^*_l\leq t^*_i$ for $l\leq i$ (constraint (\ref{constraint:monotone})) implies that
        $$
        \sum_{i=1}^{l-1} \left[\min\{r^*_{i,l}, p^*_i\} - d^*_i\right]^+   +   \sum_{i = l}^{k} \left[\min\{t^*_l, p^*_i\} - d^*_i\right]^+   \geq
        \sum_{i=1}^{l-1} z^*_i \left[r^*_{i,l} - d^*_i\right]^+   +   \sum_{i=l}^{k} z^*_i \left[t^*_l - d^*_i\right]^+$$
        which shows feasibility of $\sol_1$.

        \textbf{Step 2 --- discretization}

        Take $M = N \cdot \lceil \max\limits_{i:z^*_i > 0}{\frac{1}{z^*_i}} \rceil$, where $N = \lceil\frac{\lambda_f f^*}{\epsilon}\rceil$. Note that the value of $M$ depends on $\sol_1$.

        Define now program $P_2(k, M)$, by adding to the program $P_1(k)$ the constraints $z_i \in \{0, \frac{1}{M}, \frac{2}{M}, \dots, \frac{M}{M}\}$ for each $i$. We construct a solution to this program in the following way. Take $z'_i = \frac{\lceil z^*_i \cdot M \rceil}{M}$ and $f' = \frac{N+1}{N}f^*$.
        Let now $\sol_2 = (t^*,d^*,r^*,z',f')$.

        First, we claim that $\sol_2$ is feasible to program $P_2(k,M)$.
        To see this, observe that $z^*_i \leq z'_i \leq \frac{N+1}{N} z^*_i$ and $z' \in \{0, \frac{1}{M}, \frac{2}{M}, \dots, \frac{M}{M}\}$. The feasibility follows from multiplying both sides of the constraint~(\ref{constraint:opening2}) by $\frac{N+1}{N}$.

        Second, we claim that $v(\sol_1, P_1(k)) \leq v(\sol_2, P_2(k,M)) + \epsilon$. We have the following:
        \begin{align*}
            v(\sol_1, P_1(k)) &= \frac{\sum_i{d^*_i + z^*_i (t^*_i - d^*_i)} - \lambda_f f^*}{\sum_i{d^*_i}} \\
            &\leq \frac{\sum_i{d^*_i + z'_i (t^*_i - d^*_i)} - \frac{N+1}{N}\lambda_f f^* + \frac{1}{N}\lambda_f f^*}{\sum_i{d^*_i}} \\
            &\leq \frac{\sum_i{d^*_i + z'_i (t^*_i - d^*_i)} - \lambda_f f' + \epsilon}{\sum_i{d^*_i}} \\
            &\leq v(\sol_2, P_2(k,M)) + \epsilon
        \end{align*}
        where in the last line we use the fact that for the normalization, the denominator $\sum_i{d_i}$ can be fixed to be equal 1.

        \textbf{Finishing the proof}
        
        Define now $m_i = M \cdot z'_i$ and consider program $\hat{P}(k,m)$. Note, that all the $m_i$ variables are natural numbers as required.
        It remains to construct the solution $\overline\sol_{k,m}$ for $\hat{P}(k,m)$.
        Let $\hat{f} = M \cdot f'$ and $\overline\sol_{k,m} = (t^*,d^*,r^*,\hat{f})$.
        To see that the constraint~(\ref{constraint:opening-final}) is satisfied, multiply by $M$ both sides of valid constraint~(\ref{constraint:opening2}) for $\sol_2$
        (i.e. $\sum_{i=1}^{l-1} z'_i \left[r^*_{i,l} - d^*_i\right]^+   +   \sum_{i=l}^{k} z'_i \left[t^*_l - d^*_i\right]^+   \leq f'$).
        
        It remains to bound the value of $\sol_2$ with the value of $\overline\sol_{k,m}$:
        \begin{align*}
        v(\sol_2, P_2(k,m)) &= \frac{\sum_i{d^*_i + z'_i (t^*_i - d^*_i)} - \lambda_f f'}{\sum_i{d^*_i}} \\
        &= \frac{\sum_i{d^*_i M + z'_i M (t^*_i - d^*_i)} - \lambda_f Mf'}{\sum_i{Md^*_i}} \\
        & = \frac{\sum_i{d^*_i M + m_i (t^*_i - d^*_i)} - \lambda_f Mf'}{\sum_i{Md^*_i}} \\
        & = \frac{\sum_i{m_i t^*_i} - \lambda_f Mf' + \sum_i{(M-m_i)d^*_i}}{\sum_i{m_i d^*_i} + \sum_i{(M-m_i) d^*_i}} \\
        & \leq \frac{\sum_i{m_i t^*_i} - \lambda_f Mf'}{\sum_i{m_i d^*_i}} \\
        &= v(\overline\sol_{k,m}, \hat{P}(k,m))
        \end{align*}
        where the last inequality follows from the fact that the nominator is larger than the denominator (as this fraction gives an upper bound on approximation factor which must be greater than 1).
    \qed\end{proof}

%% file: combining.tex
%Show how to combine JMS with LP rounding.

By Corollary~\ref{corollary:bifactor}, the adapted JMS algorithm is a $(1.11,1.78)$-approximation algorithm for the \textsc{ufl} problem with penalties.

It remains to note that the applicability of the LP-rounding algorithms for \textsc{ufl} to the \textsc{flp} problem has already been studied. In particular the algorithm 4.2 of~\cite{LiDXX15} is an adaptation of the LP rounding algorithms for \textsc{ufl} by Byrka and Aardal~\cite{byrka2010optimal} and Li~\cite{li20111} to the setting with penalties.

Note also that Qiu and Kern~\cite{qiu2016factor} made an attempt on finalising the work on \textsc{ufl} with penalties and analysing the adapted JMS algorithm, and correctly argued that once the analogue of the JMS algorithm for the penalty version of the problem is known the improved approximation ratio for the variant with penalties will follow. Their analysis of the adapted JMS was incorrect (and the paper withdrawn from arxiv). By providing the missing analysis of the adapted JMS, we fill in the gap and obtain:

\begin{corollary}
 There exists a $1.488$-approximation algorithm for the \textsc{flp}.
\end{corollary}

\begin{corollary}
 For any $\lambda_f \geq 1.6774$ and $\lambda_{c+p} = 1+ \frac{2}{e^{\lambda_f}}$ there exists a bifactor $(\lambda_f, \lambda_{c+p})$-approximation algorithm for \textsc{flp}.
\end{corollary}

%% file: reduction.tex
%First, we define the Facility Location with per-client Monotone Concave Connection Costs (\textsc{ncc-fl}) problem. In this setting we have a standard metric uncapacitated facility location instance with clients $D$, facilities $F$ and distances $d_{j,i}$ satisfying triangle inequality. However, now for every client $j \in D$ we have an increasing concave function $g_j$ which maps distances to costs. Therefore the connection cost of client $j$ to facility $i$ is equal to $g_j(d_{j,i})$.

We will now discuss how to use algorithms for the  \textsc{flp} problem to solve \textsc{ncc-fl} problem. 
To this end we introduce yet another variant of the problem: facility location with penalties and multiplicities \textsc{flpm}. In this setting each client $j$ has two additional parameters: penalty $p_j$ and multiplicity $m_j$, both being nonnegative real numbers. If client $j$ is served by facility $i$ the service cost is $m_j \cdot d_{ij}$ and if it is not served by any facility the penalty cost is $m_j \cdot p_j$. 

%The main idea in the reduction it to create multiple copies of each client with the different penalties in such a way that the additional cost of applying a function $g$ (i.e. $g_j(d_{i,j}) - d_{i, j}$) is paid by penalties.
\begin{lemma}
\label{lem:ncc-fl}
    There is an approximation preserving reduction from  \textsc{ncc-fl} to \textsc{flpm}.
\end{lemma}
\begin{proof}
	Take an instance $I=(F, D, d, g)$ of the \textsc{ncc-fl} problem with $|F| = n$ facilities.  Create the instance $I'=(F, D',d', p, m)$ as follows. The set of the facilities is the same as in original instance. For each client $j \in D$, we will create in $I'$ multiple copies of $j$.
	
	Fix a single client $j$. Sort all the facilities by their distance to $j$ and let 
	$d^{(j)}_{1} \leq d^{(j)}_{2} \leq \dots \leq d^{(j)}_{n}$ be the sorted distances.
	For every $k \in [n-1]$ define also
	\begin{align}
		m^{(j)}_k = \frac{g_j(d^{(j)}_{k}) - g_j(d^{(j)}_{k-1})}{d^{(j)}_{k} - d^{(j)}_{k-1}} - \frac{g_j(d^{(j)}_{k+1}) - g_j(d^{(j)}_{k})}{d^{(j)}_{k+1} - d^{(j)}_{k}}
		\label{m-definition-1}
	\end{align} where for convenience we define $d^{j}_0 = g_j(d^{(j)}_0) = 0$. Let also
	\begin{align}
		m^{(j)}_n = \frac{g_j(d^{(j)}_{n}) - g_j(d^{(j)}_{n-1})}{d^{(j)}_{n} - d^{(j)}_{n-1}}
		\label{m-definition-2}
	\end{align}

	Observe that concavity of $g_j$ implies that every $m^{(j)}_k$ is nonnegative. Now, for each $k \in [n]$, we create a client $j_k$ in the location of $j$ with penalty set to $p_{j_k}=d^{(j)}_k$ and multiplicity $m^{(j)}_k$. It remains to show that for any subset of facilities $F' \subseteq F$ it holds that $cost_{I}(F') = cost_{I'}(F')$, where $cost_{.}(F')$ denotes the cost of a solution obtained by optimally assigning clients to facilities in $F'$. To see this, consider a client $j \in D$ and its closest facility $f$ in $F'$. Let $k$ be the index of $f$ in the vector $d^{(j)}$ of sorted distances, i.e. the distance between $j$ and $f$ is $d^{(j)}_k$. Observe that for the solution $F'$ in $I'$ the clients $j_1, \dots, j_i, \dots, j_{k-1}$ pay their penalty $d^{(j)}_i$ (as the penalty is at most the distance to the closest facility) and the clients $j_k, \dots j_{n}$ can be connected to $f$.
	Therefore, we require that for any $k$,
	
	\begin{align}
		g_j(d^{(j)}_k) = \sum_{i=1}^{k-1}m^{(j)}_{i} d^{(j)}_i + \sum_{i=k}^{n}m^{(j)}_{i} d^{(j)}_k
		\label{m-g-linear-system}
	\end{align}
	It can be observed, that $m^{(j)}_{i}$ as defined in (\ref{m-definition-1}) and (\ref{m-definition-2}) satisfy the system of equalities (\ref{m-g-linear-system}).
\qed\end{proof}

It remains to show how to solve \textsc{flpm}. Recall that our algorithm for \textsc{flp} is a combination of the JMS algorithm and an LP rounding algorithm. We will now briefly argue that each of the two can be adapted to the case with multiplicities. 

To adapt the JMS algorithm, one needs to take multiplicities into account when calculating the contributions of individual clients towards facility opening. Then in the analysis multiplicities can be scaled up and discretized to lose only an epsilon factor. Here we utilize that non-polynomial blowup in the number of clients is not a problem in the analysis.

To adapt the LP rounding algorithm, we first observe that multiplicities can easily be introduced to the LP formulation, and hence solving an LP relaxation of the problem with multiplicities is not a problem. Next, we utilize that the analysis of the expected connection (and penalty) cost of the algorithm is a per-client analysis. Therefore, by the linearity of expectation combined with linearity of the objective function with respect to multiplicities, the original analysis applies to the setting with multiplicities. A similar argument was previously utilized in~\cite{byrka2018proportional}.

%% file: IRP.tex
In this section we give the improved algorithms for all the three variants of  the Star Inventory Routing Problem with Facility Location that was recently introduced by Jiao and Ravi~\cite{IRP:WADS}. First, we recall the definition. We are given the set of facilities $F$ and clients $D$ in the metric space ($d$), a time horizon $1, \dots, T$, a set of demand points $(j,t)$ with $u^j_t$ units of demand requested by client $j$ due by day $t$, facility opening costs $f_i$, holding costs $h^j_{s,t}$ per unit of demand delivered on day $s$ serving ($j,t$). The objective is to open a set of facilities and plan the deliveries to each client to minimize the total facility opening costs, client-facility connections and storage costs. Consider first the uncapacitated variant in which every delivery can contain unbounded number of goods.

Observe that once the decision which facilities to open is made, each client can choose the closest open facility and use it for all the deliveries. In that case, we would be left with a single-item lot-sizing problem which can be solved to optimality.
The above view is crucial for our approach. Observe that we can precompute all the single-item lot-sizing instances for every pair $(i,j)$. Now we are left with a specific facility location instance that is captured by \textsc{ncc-fl}. The following lemma proves this reduction.
\begin{lemma}
\label{lem:sirpfl}
    There is a 1.488 approximation algorithm for uncapacitated \textsc{sirpfl}.
\end{lemma}
\begin{proof} [Proof of Lemma~\ref{lem:sirpfl}]
    For each $j \in D$ and $i \in F$ solve the instance of a single-item lot-sizing problem with delivery cost $d_{j,i}$ and demands and holding costs for client $j$ to optimality~\cite{wagner1958dynamic}. Now define $g_j(d_{j,i})$ to be the cost of this computed solution and linearly interpolate other values of function $g$.

    It is now easy to see that $g_j$ is an increasing concave function. This follows from the fact, that the optimum solution to the lot-sizing problem with delivery cost $x$ is also a feasible solution to the problem with delivery cost $\alpha \cdot x$ for $\alpha \geq 1$. Moreover the value of this solution for increased delivery cost is at most $\alpha$ times larger.

    In this way we obtained the instance of \textsc{ncc-fl} problem which can be solved using the algorithm of Lemma~\ref{lem:ncc-fl}. Once we know which facilities to open, we use optimal delivery schedules computed at the beginning.
\qed\end{proof}
Jiao and Ravi studied also capacitated splittable and capacitated unsplittable variants obtaining $24$ and $48$ approximation respectively~\cite{IRP:WADS}. By using corresponding $3$ and $6$ approximation for the capacitated splittable and unsplittable Inventory Access Problem (\textsc{IAP}) given in~\cite{IRP:WADS} (the variants of single-item lot-sizing) and a similar reduction to \textsc{ncc-fl} as above while using a suitable bi-factor algorithm for \textsc{FLP} we are able to give improved approximation algorithms for both capacitated variants of \textsc{sirpfl}.

\begin{lemma}
\label{lem:capacitated-sirpfl}
    There is a $3.236$-approximation algorithm for capacitated splittable \textsc{sirpfl} and $6.029$-approximation algorithm for capacitated unsplittable \textsc{sirpfl}.
\end{lemma}
\begin{proof} [Proof of Lemma~\ref{lem:capacitated-sirpfl}]
    The approach is the same as in proof of Lemma~\ref{lem:sirpfl} but with a little twist. We give details only for splittable case as the unsplittable variant follows in the same way.

    For each $j \in D$ and $i \in F$ run the $3$-approximation algorithm~\cite{IRP:WADS} for the instance of a corresponding splittable Inventory Access Problem problem with delivery cost $d_{j,i}$ and demands and holding costs for client $j$.

    Notice that we cannot directly define $g_j(d_{j,i})$ to be the cost of the computed solution as the resulting function would not necessarily be concave (due to using approximate solutions instead of optimal).

    Therefore, we construct $g_j$ for each $j \in D$ in a slightly different way. W.l.o.g assume that $d_{j,1} \leq d_{j,2} \dots \leq d_{j,n}$. Let also $A(x)$ be the computed $3$-approximate solution for \text{IAP} with delivery cost $x$ and let $V(S, x)$ be the value of solution $S$ for \text{IAP} with delivery cost $x$. Notice that for $x<y$, the solution $A(x)$ is feasible for the same \text{IAP} instance but with delivery cost $y$. In particular, the following bound on cost is true: $V(A(x), y) \leq \frac{y}{x} \cdot V(A(x), x)$.

    We now construct a sequence of solutions. Let $S_1 = A(d_{j,1})$. Now, for each $i \in [n-1]$ define:
$$
    S_{i+1} =
    \begin{cases}
        S_i,            & \text{if } V(S_i,d_{j,i+1}) < V(A(d_{j,i+1}),d_{j,i+1})\\
        A(d_{j,i+1}),   & \text{otherwise}
    \end{cases}
$$
	Finally take $g(d_{j,i}) = V(S_i, d_{j,i})$ and linearly interpolate other values. It can be easily observed that $g_j$ is a nondecreasing concave function.

    Finally, we are using the bifactor ($\lambda_f, 1+2e^{-\lambda_f}$)-approximation algorithm to solve the resulting instance of \textsc{ncc-fl}. Because we also lose a factor of $3$ for connection cost, the resulting ratio is equal to $\max\{\lambda_f, 3\cdot(1+2e^{-\lambda_f})\}$.
    The two values are equal for $\lambda_f \approx 3.23594$.
\qed\end{proof}